\documentclass[11pt,a4paper]{article}

\usepackage{graphicx}
\usepackage{amssymb}
\usepackage{amsfonts}
\usepackage{amsbsy}
\usepackage{amsmath}
\usepackage{amsthm}
\usepackage{geometry}

\usepackage{tikz}              
\usetikzlibrary{arrows.meta}  
\usetikzlibrary{positioning}  
      
\usepackage{subcaption}   
\usepackage[ruled]{algorithm2e}

\newcommand{\beq}{\begin{equation}}
\newcommand{\eeq}{\end{equation}}

\newcommand{\bmat}{\left[ \begin{array}}
\newcommand{\emat}{\end{array} \right]}

\newcommand{\rank}{\mbox{rank}}

\newtheorem{theorem}{Theorem}[section]
\theoremstyle{definition}
\newtheorem{definition}{Definition}[section]
\newtheorem{lemma}{Lemma}[section]
\newtheorem{example}{Example}[section]

\title{Error-correction Decoding for 2D Convolutional Codes with a Generalized Viterbi Algorithm}
\usepackage{xcolor}
\usepackage{authblk}
\usepackage{comment}
\author[1]{Zita Abreu}
\author[2]{Elisa Junghans}
\author[3]{Julia Lieb}
\author[4]{Raquel Pinto}
\affil[1]{CIDMA, University of Madeira, Portugal\thanks{zita.abreu@staff.uma.pt}}
\affil[2]{TU Ilmenau, Germany\thanks{elisa.junghans@tu-ilmenau.de}}
\affil[3]{University of Bayreuth, Germany\thanks{Julia.Lieb@uni-bayreuth.de}}
\affil[4 ]{CIDMA, University of Aveiro, Portugal\thanks{raquel@ua.pt}}

\begin{document}

\maketitle

\begin{abstract}
Two-dimensional (2D) convolutional codes are a generalization of one-dimensional (1D) convolutional codes and are naturally suitable for encoding two-dimensional data. Although their theoretical properties are well established, practical applications are limited due to the lack of error-correction decoding algorithms. While the Viterbi algorithm provides an efficient decoding method for 1D convolutional codes, no analogous approach exists for the 2D case. In this paper, we introduce a decoding algorithm for 2D convolutional codes, called the 2D Viterbi algorithm, extending the principles of the classical Viterbi algorithm to the two-dimensional setting. The presented 2D Viterbi algorithm works with linear systems representations of 2D convolutional codes and its complexity depends on the dimension of the corresponding 2D linear system. Hence, we also provide several algorithms to obtain 2D realizations of small dimension.
\end{abstract}

\medskip


\medskip

\section{Introduction}
Two-dimensional (2D) convolutional codes are a natural generalization of convolutional codes, to which we will refer in this paper as one-dimensional
(1D) convolutional codes. 2D convolutional codes deal with data represented by two-dimensional sequences and therefore are naturally fitting to deal with two-dimensional information, like pictures, data storage, etc.

The theory of 2D convolutional codes is very well understood and has been investigated by many authors in the last years. These codes have been introduced in \cite{FoVa94,valcherF94}. In these papers the algebraic properties of 2D convolutional codes are studied in detail. In \cite{ClNaPePi12}, the authors investigated the distance properties of 2D convolutional codes of rate $1/n$ and obtained an upper bound (called the 2D generalized Singleton bound) for these codes and presented  constructions of 2D convolutional codes with optimal distance. In \cite{ClNaPePi16}, the authors generalized these concepts and results for 2D convolutional codes of any rate and presented constructions of such codes for certain parameters. In \cite{AlNaPi15}, further constructions were presented for a more general set of parameters. Input-state-output representations of 2D systems were studied in \cite{NaPePi10}.

Even though the theory of 2D convolutional codes is very well developed, these codes are not used in practical applications. One of the main reasons is that there exists no decoding algorithm for error-correction with these codes. In \cite{ClNaPiSi16} and in \cite{LiPi23}, decoding algorithms for 2D convolutional codes were presented for an erasure channel, i.e., for a channel in which a symbol arrives correctly or does not arrive at all. However, when we consider channels in which errors occur, there is no procedure to detect and correct errors that appear during the transmission of the information.

The most famous error-correction decoding algorithm for 1D convolutional codes is the Viterbi algorithm. This algorithm decodes a received word after channel transmission by considering all codewords of the code up to a certain length and then selecting the codeword which is closest to the received word. This algorithm applies techniques from dynamic programming to the linear systems representation of the convolutional code. In this paper we define a similar error-correction decoding algorithm for 2D convolutional codes, which we will call the 2D Viterbi algorithm.
The computational complexity of this algorithm depends on the dimension of the underlying 2D linear system. Finding a realization for a 2D encoder of minimal dimension is a challenging open problem. However, in this paper, we provide several strategies to obtain 2D realizations of small dimensions.

The paper is organized as follows. Section 2 presents the basic background on 2D convolutional codes including system representations of 2D convolutional codes. In Section 3, we introduce a decoding method for 2D convolutional codes, called the 2D Viterbi algorithm, which extends the ideas of the classical Viterbi algorithm to the two-dimensional case. In Section 4, we analyze the complexity of the 2D Viterbi algorithm and discuss ways to obtain an encoder realization of small dimension in order to decrease this decoding complexity.

\section{Two-dimensional convolutional codes}\label{systems}

We denote by $\mathbb F$ any finite field. Two-dimensional convolutional codes deal with data represented by a two-dimensional (2D) compact support sequence
$$
\begin{array}{cccc}
w: & \mathbb N_{0}^2 & \rightarrow & \mathbb F^n \\
& (i,j) & \mapsto w_{ij}.
\end{array}
$$
Such 2D sequences can be written as polynomials in the two indeterminates $z_1$ and $z_2$, $$w(z_1,z_2)=\displaystyle \sum_{(i,j) \in \mathbb N_0^2} w_{ij} z_1^i
z_2^j$$
with $w_{ij}\in\mathbb F$.
Let us denote by $\mathbb F[z_1,z_2]$ the ring of polynomials in $z_1$ and $z_2$.

\begin{definition}
A two-dimensional (2D) $(n,k)$ convolutional code $\cal C$ is a free $\mathbb{F}[z_1,z_2]$-submodule of $\mathbb{F}[z_1,z_2]^n$ of rank $k$. A matrix $G(z_1,z_2) \in \mathbb F[z_1,z_2]^{n \times k}$ whose columns form a basis of $\cal C$ is called an encoder of $\cal C$ and it follows that
\begin{eqnarray*}
{\cal C} & = & \text{Im}_{\mathbb{F}[z_1,z_2]}G(z_1,z_2)\\
& = & \{G(z_1,z_2)u(z_1,z_2) \, | \, u(z_1,z_2) \in \mathbb F[z_1,z_2]^k\}.
\end{eqnarray*}
\end{definition}
Given a polynomial matrix $G(z_1,z_2) \in \mathbb F[z_1,z_2]^{n \times k}$ we define the $i$-th column degree of $G(z_1,z_2)$ as the maximum degree of the entries of the $i$-th column of $G(z_1,z_2)$ where the (total) degree of a polynomial in two indeterminates $p(z_1,z_2) = \displaystyle \sum_{(i,j) \in \mathbb N_0^2} p_{ij}z_1^iz_2^j$, with $p_{ij} \in \mathbb F$ is defined as usual as $\deg(p(z_1,z_2))= \mbox{max}\{i+j : p_{ij} \neq 0\}$.

A unimodular matrix $U(z_1,z_2) \in \mathbb F[z_1,z_2]^{k \times k}$ is a polynomial matrix that admits a polynomial inverse matrix. $U(z_1,z_2)$ is unimodular if and only if $\det(U(z_1,z_2)) \in \mathbb{F} \backslash \{0\}$; see \cite{levy81}. Two encoders $G(z_1,z_2), \tilde G(z_1,z_2) \in \mathbb F[z_1,z_2]^{n \times k}$ of a 2D convolutional code differ by right multiplication with a unimodular matrix, i.e., $G(z_1,z_2)=\tilde{G}(z_1,z_2)U(z_1,z_2)$, for some unimodular matrix $U(z_1,z_2) \in \mathbb F[z_1,z_2]^{k \times k}$. \\

Convolutional codes can be represented by linear systems.
 A 2D linear system models relations between 2D data.  There are several state-space models of a 2D linear system. In the Fornasini-Marchesini state-space model (see \cite{FoMa86}) the data relation is given by the updating equations
\begin{align} \label{system}
x(i\!+\!1,j\!+1\!) & = A_1 x(i,j\!+\!1) + A_2 x(i\!+\!1,j) + B_1 u(i,j\!+\!1) + B_2 u(i\!+\!1,j) \nonumber\\
w(i,j) & =  C x(i,j) + D u(i,j),
\end{align}
\noindent where $A_1, A_2 \in \mathbb F^{s \times s}$, $B_1, B_2 \in \mathbb F^{s\times k}$, $C \in \mathbb F^{n \times s}$, $D \in \mathbb F^{n \times k}$, $s, n, k \in \mathbb{N}_{0}$, and the system is denoted by $\Sigma\!=\!\left(A_1,A_2,B_1,B_2,C,D\right)$.  We say that $\Sigma$ has dimension $s$. The vectors $x(i,j)\in \mathbb F^s$, $u(i,j)\in \mathbb F^k$ and $w(i,j)\in \mathbb F^n$ represent the local state, input and output at $(i,j)$, respectively.

We can define sets that partition the plane $\mathbb Z \times \mathbb Z$ in two regions. The $i$-th separation set of $\mathbb Z \times \mathbb Z$ is defined as 
$$
{\cal S}_i=\{(r,t)\in\mathbb Z\times\mathbb Z \, : \, r+t=i\}
$$
and the $i$-th global state as
$$
\begin{array}{cccc}
{\cal X}_i: & \mathbb Z & \rightarrow & \mathbb F^s \\
& r & \mapsto & x(r,i-r)
\end{array}
$$
constituted by the sequence of local states with support in ${\cal S}_i$.
We call the future of a point $(r,t)$ on $\mathbb Z \times \mathbb Z$ the set
$$
{\cal F}_{(r,t)}=\{(i,j) \in \mathbb Z \times \mathbb Z \, :  r \leq i, t \leq j\}
$$
and the future region with respect to ${\cal S}_i$ is 
$$
{\cal F}_{{\cal S}_i}=\bigcup_{(r,t) \in {\cal S}_i}{\cal F}_{(r,t)}.
$$

Given a system $\Sigma\!=\!\left(A_1,A_2,B_1,B_2,C,D\right)$ and considering ${\cal X}_0=\bf{0}$, it follows that for any input $u(z_1,z_2)$ with support in ${\cal F}_{{\cal S}_0}$ the corresponding state is given by 
$$
x(z_1,z_2)=(I-A_1z_1-A_2z_2)^{-1}(B_1z_1+B_2z_2)u(z_1,z_2)
$$
and the corresponding output is given by
$$
w(z_1,z_2)=\left(C(I-A_1z_1-A_2z_2)^{-1}(B_1z_1+B_2z_2)+D \right)u(z_1,z_2).
$$
The matrix $T(z_1,z_2)= C(I-A_1z_1-A_2z_2)^{-1}(B_1z_1+B_2z_2)+D$ is called the transfer matrix of $\Sigma$; see \cite{FoMa86}.
We say that a 2D system $\Sigma\!=\!\left(A_1,A_2,B_1,B_2,C,D\right)$ is a realization of an encoder $G(z_1,z_2)$ of a 2D convolutional code $\cal C$ if $$G(z_1,z_2)= C(I-A_1z_1-A_2z_2)^{-1}(B_1z_1+B_2z_2)+D$$ and the codewords of $\cal C$ are the outputs of $\Sigma$ corresponding to the inputs $u(z_1,z_2) \in \mathbb F[z_1,z_2]^k$ and ${\cal X}_0=\bf{0}$. We call $\Sigma$ a system representation of $\cal C$.

In \cite{BiFoMa89}, the authors give a construction procedure for a realization of a rational matrix in two indeterminates. Next we give this algorithm that constructs such a realization of an encoder $G(z_1,z_2) \in \mathbb F[z_1,z_2]^{n \times k}$ of a 2D convolutional code. For that we consider for any integer $\ell \geq 1$ the constant matrices\\
\[
A_{1}^{(\ell)}=\left[\begin{matrix}
0 & \cdots & 0 & \vline & M_{\ell} & & & \\
\vdots &  & \vdots & \vline & & M_{\ell-1} & & \\
\vdots &  & \vdots & \vline   & & & \ddots & \\
0 & \cdots & 0 & \vline &  & & & M_2 \\
\hline \\
0 & \cdots &\cdots  & \cdots & \cdots & \cdots & \cdots & 0 \\
0 & \cdots &\cdots  & \cdots & \cdots & \cdots & \cdots & 0
\end{matrix}
\right] \in \mathbb F^{\frac{(\ell+3)\ell}{2} \times \frac{(\ell+3)\ell}{2}}
\]
with   $\displaystyle M_j=\left[
\begin{matrix}
0  & \cdots & 0 \\
& I_j &
\end{matrix}
\right] \in \mathbb F^{(j+1) \times j}$,
\[
A_{2}^{(\ell)}=\left[
\begin{matrix}
0 & \cdots & 0 & \vline & N_\ell & & & \\
\vdots &  & \vdots & \vline & & N_{\ell-1} & & \\
\vdots &  & \vdots & \vline   & & & \ddots & \\
0 & \cdots & 0 & \vline &  & & & N_2 \\
\hline \\
0 & \cdots &\cdots  & \cdots & \cdots & \cdots & \cdots & 0 \\
0 & \cdots &\cdots  & \cdots & \cdots & \cdots & \cdots & 0
\end{matrix}
\right] \in \mathbb F^{\frac{(\ell+3)\ell}{2} \times \frac{(\ell+3)\ell}{2}}
\]
with $\displaystyle
N_j=\left[
\begin{matrix}
1 & 0  & \cdots & 0 \\
 0&\cdots & \cdots &0\\
 \vdots &&&\vdots\\
 0&\cdots &\cdots&0
\end{matrix}
\right]  \in \mathbb F^{(j+1) \times j}$
and
\[
B_1^{(\ell)} = \left[\begin{matrix}
0 \\
\vdots \\
0 \\
\hline
0 \\
1
\end{matrix}\right]\in \mathbb F^{\frac{(\ell+3)\ell}{2}} \mbox{   and   }
B_2^{(\ell)} =\left[\begin{matrix}
0 \\
\vdots \\
0 \\
\hline
1 \\
0
\end{matrix}\right]\in \mathbb F^{\frac{(\ell+3)\ell}{2}}.
\]

\begin{algorithm}[H]
    \DontPrintSemicolon
    \caption{\cite{BiFoMa89}}\label{Alg 1}
    \KwIn{$G(z_1,z_2) \in \mathbb F[z_1,z_2]^{n \times k}$ with column degrees $\nu_1,\nu_2,\dots, \nu_k$.}
    \nl Let $D:=G(0,0)$ and $\hat G(z_1,z_2):=G(z_1,z_2)-G(0,0).$\;
    \nl Let 
    \begin{align*}
        A_1 &:={\rm diag}\{A_1^{(\nu_1)},A_1^{(\nu_2)},\dots,A_1^{(\nu_k)}\},\\
        A_2 &:={\rm diag}\{A_2^{(\nu_1)},A_2^{(\nu_2)},\dots,A_2^{(\nu_k)}\},\\
        B_1 &:={\rm diag}\{B_1^{(\nu_1)},B_1^{(\nu_2)},\dots,B_1^{(\nu_k)}\},\\
        B_2 &:={\rm diag}\{B_2^{(\nu_1)},B_2^{(\nu_2)},\dots,B_2^{(\nu_k)}\}.
    \end{align*}\;
    \nl Let $\mathcal{Z}^{(\ell)}$ be the matrix constituted by the elements $z_1^iz_2^j$ of degree smaller than or equal to $\ell$ and greater than $0$ in decreasing lexicographic order \footnotemark, i.e.,
    \[ \mathcal{Z}^{(\ell)} =\left[ \begin{array}{ccccccccccccc} z_2^{\ell} &  z_1z_2^{\ell-1} & \hspace{-1mm}\cdots \hspace{-1mm}& z_1^{\ell} &  z_2^{\ell-1} &  z_1z_2^{\ell-2} & \hspace{-1mm}\cdots\hspace{-0.5mm} & z_1^{\ell-1} & \hspace{-1mm}\cdots\hspace{-0.5mm} & z_2 & z_1 \end{array} \right]^T. \]\;
    \nl Define $$\Psi(z_1,z_2) ={\rm diag}\{\mathcal{Z}^{(\nu_1)},\mathcal{Z}^{(\nu_2)},\dots,\mathcal{Z}^{(\nu_k)}\}$$ and let $C \in \mathbb F^{n \times (\sum_{i=1}^k \frac{\nu_i(\nu_i+3)}{2})}$ be such that $\hat G(z_1,z_2)=C\Psi(z_1,z_2)$.
    In case $\nu_i=0$ for some $i$, the $i$-th columns in $B_{1}$, $B_{2}$ and $\Psi(z_1,z_2)$ are equal to zero, and the $i$-th diagonal blocks $A_{1}^{(\nu_i)}$ and $A_{2}^{(\nu_i)}$ are empty.\;
    \KwOut{$\Sigma\!=\!\left(A_1,A_2,B_1,B_2,C,D\right)$}
\end{algorithm}
\footnotetext{{Given $z_1^{i_1} z_2^{j_1}$ and $z_1^{i_2} z_2^{j_2}$, we say that $z_1^{i_1} z_2^{j_1} < z_1^{i_2} z_2^{j_2}$ if $i_1+j_1 < i_2+j_2$ or, in case $i_1 + j_1=i_2 +j_2$, if $i_1<i_2$}}

\begin{theorem}\label{th realization}
    Algorithm \ref{Alg 1} gives a realization of  $G(z_1,z_2) \in \mathbb F[z_1,z_2]^{n \times k}$. 
\end{theorem}
\begin{proof}
The theorem follows from the fact that 
$$
\Psi(z_1,z_2)=(I-A_1z_1-A_2z_2)^{-1}(B_1z_1+B_2z_2).
$$
\end{proof}

Let $G(z_1,z_2) \in \mathbb F[z_1,z_2]^{n \times k}$ and let $\Sigma\!=\!\left(A_1,A_2,B_1,B_2,C,D\right)$ be the realization of $G(z_1,z_2)$ of dimension $s$ produced by Algorithm 1. It is easy to see that\\ det$\left[\begin{array}{c} I-A_1z_1-A_2z_2 \end{array}\right]=1$, which means that $ I-A_1z_1-A_2z_2 $ is unimodular. Thus, a polynomial input $u(z_1,z_2) \in \mathbb F[z_1,z_2]^k$ produces the  state $x(z_1,z_2)=(I-A_1z_1-A_2z_2)^{-1}(B_1z_1+B_2z_2)u(z_1,z_2)$ which is also polynomial, i.e., which belongs to $\mathbb F[z_1,z_2]^s$. Thus, if $G(z_1,z_2) \in \mathbb F[z_1,z_2]^{n \times k}$ is an encoder of a 2D convolutional code and $\Sigma\!=\!\left(A_1,A_2,B_1,B_2,C,D\right)$ is the realization of $G(z_1,z_2)$ given by Algorithm 1, we have that the codewords of $\cal C$ correspond to the outputs of the system with corresponding inputs and states having finite support in $\mathbb{N}_{0}^2$.

Another way of seeing the evolution of the trajectories of a 2D system is by considering its global states, inputs and outputs. The $i$-th global input is defined as
$$
\begin{array}{cccc}
{\cal U}_i: & \mathbb Z & \rightarrow & \mathbb F^k\\
& r & \mapsto & u(r,i-r)
\end{array}
$$
and the $i$-th global output as
$$
\begin{array}{cccc}
{\cal W}_i: & \mathbb Z & \rightarrow & \mathbb F^n\\
& r & \mapsto & w(r,i-r)
\end{array}
$$
constituted by the sequences of local inputs and outputs with support in ${\cal S}_i$ \cite{FoMa86}.
We can represent the $i$-th global state, input and output states as
$$
X_i(\xi)=\sum_{j \in \mathbb Z} x(i-j,j)\xi^j,
$$
$$
U_i(\xi)=\sum_{j \in \mathbb Z} u(i-j,j)\xi^j,
$$
$$
W_i(\xi)=\sum_{j \in \mathbb Z} w(i-j,j)\xi^j,
$$
and the sequences of (local) states, inputs and outputs, $x(z_1,z_2)$, $u(z_1,z_2)$ and $w(z_1,z_2)$, as
$$
\sum_{i \in \mathbb Z} X_i(\xi) \eta^i,
$$
$$
\sum_{i \in \mathbb Z} U_i(\xi) \eta^i,
$$
$$
\sum_{i \in \mathbb Z} W_i(\xi) \eta^i,
$$
which corresponds to the change of variables $\xi=z_1^{-1}z_2$ and $\eta=z_1$.

The global updating equations of a 2D system $\Sigma\!=\!\left(A_1,A_2,B_1,B_2,C,D\right)$ are
\begin{eqnarray*}
X_{i+1}(\xi) & = & (A_1+A_2 \xi)X_{i}(\xi) + (B_1+B_2 \xi) U_i(\xi) \\
W_i(\xi) & = & C X_i(\xi) + D U_i(\xi) .
\end{eqnarray*}

Let $\Sigma$ be a realization of dimension $s$ obtained from Algorithm 1 of an encoder $G(z_1,z_2) \in \mathbb F[z_1,z_2]^{n \times k}$ of a 2D convolutional code. Note that the i-th global states, inputs and outputs corresponding to trajectories whose output are the codewords of the code, are polynomials of degree smaller than or equal to $i$ with coefficients in $\mathbb F^{s}$, $\mathbb F^{k}$ and $\mathbb F^{n}$, respectively. Let us denote such sets by $\mathbb F[\xi]^s_{\leq i}$, $\mathbb F[\xi]^k_{\leq i}$ and $\mathbb F[\xi]^n_{\leq i}$, respectively.

\section{Two-dimensional Viterbi algorithm}\label{decoding}

Let $\cal C$ be an $(n,k)$ 2D convolutional code and let $$r(z_1,z_2) = \displaystyle \sum_{(i,j) \in \mathbb N_0^2} r_{ij} z_1 ^i z_2 ^j \in \mathbb F[z_1,z_2]^n$$ be a received word after channel transmission. In this section, we present a decoding procedure to obtain the closest codeword of $\cal C$ to $r(z_1,z_2)$, i.e., the codeword $w(z_1,z_2)= \displaystyle \sum_{(i,j) \in \mathbb N_0^2} w_{ij} z_1 ^i z_2 ^j \in \cal C$ with minimal distance to $r(z_1,z_2)$, where the distance between $w(z_1,z_2)$ and $r(z_1,z_2)$ is defined as $$d(w(z_1,z_2),r(z_1,z_2))=\displaystyle \sum_{(i,j) \in \mathbb N_0^2}d(w_{ij},r_{ij})$$ and $d(w_{ij},r_{ij})$ is the Hamming distance between $w_{ij}$ and $r_{ij}$, which is the number of positions in which $w_{ij}$ and $r_{ij}$ differ, or equivalently, the number of nonzero entries of the vector $w_{ij} - r_{ij}$ (the distance between two words in one indeterminate 
is defined in the same way). For this scenario we will adapt the Viterbi algorithm \cite{Viterbi1967}, which is the most popular error-correction decoding algorithm for 1D convolutional codes. 

The Viterbi algorithm works as follows. When a word transmitted through the communication channel is received, the Viterbi algorithm considers all codewords, until a certain length, and compares the received word to each of them.  Then it selects a codeword which has minimum distance to the received word. 

Let $\cal C$ be a 2D  convolutional code, $G(z_1,z_2) \in \mathbb F[z_1,z_2]^{n \times k}$ an encoder of $\cal C$ and $\Sigma\!=\!\left(A_1,A_2,B_1,B_2,C,D\right)$ 
a realization of dimension $s$, for example the one given by Algorithm 1, with global updating equations
\begin{eqnarray*}
X_{i+1}(\xi) & = & (A_1+A_2 \xi)X_{i}(\xi) + (B_1+B_2 \xi) U_i(\xi)\\
W_i(\xi) & = & C X_i(\xi) + D U_i(\xi) .
\end{eqnarray*}
The evolution of the trajectories of the system (that will generate the codewords of the code) can be represented by a trellis diagram, as follows: 
\begin{enumerate}
\item for each $i \in \mathbb{N}_0$ consider all the elements of $\mathbb F[\xi]^s_{\leq i}$ as the nodes at step $i$; 
\item for each $i \in \mathbb N_0$, the nodes $X_i(\xi) \in \mathbb F[\xi]^s_{\leq i}$ and $X_{i+1}(\xi) \in \mathbb F[\xi]^s_{\leq i+1}$ are joined by an arc (starting in $X_i(\xi)$ and ending in $X_{i+1}(\xi)$), if there exists $U_i(\xi) \in \mathbb F[\xi]^k_{\leq i}$ such that 
$$
X_{i+1}(\xi) = (A_1+A_2 \xi)X_i(\xi) + (B_1+B_2 \xi) U_i(\xi).
$$
The arc $(X_i(\xi),X_{i+1}(\xi))$ is labeled with $(U_i(\xi),W_i(\xi))$, where
$$
W_i(\xi)  =  C X_i(\xi) + D U_i(\xi)
$$
is an element of $\mathbb F[\xi]^n_{\leq i}$ and we say that $X_i(\xi)$ is a predecessor of $X_{i+1}(\xi)$.
\end{enumerate}

The codewords of $\cal C$ correspond to the paths in the trellis diagram that start at $X_0(\xi)=\bf{0}$ at step $0$ and end at the zero state $\bf{0} \in \mathbb F[\xi]^s_{\leq i}$ at step $i$, for some $i \in \mathbb N_0$. If 
$$(U_0(\xi),W_0(\xi)), (U_1(\xi),W_1(\xi)), \dots, (U_i(\xi),W_i(\xi))
$$ is the sequence of the labels of such a path, the corresponding codeword is equal to $w(z_1,z_2)=W_0(\xi) + W_1(\xi)\eta + \cdots + W_i(\xi)\eta^{i}$,
where $$\xi=z_1^{-1}z_2 \text{ and } \eta=z_1.$$
Given a received word $r(z_1,z_2) \in \mathbb F[z_1,z_2]^n$ and considering the same change of variables $\xi=z_1^{-1}z_2$ and $\eta=z_1$, we can write $r(z_1,z_2)$ as $\sum_{i \in \mathbb N_0}R_{i}(\xi)\eta^i$, with $R_{i}(\xi) \in \mathbb F[\xi]_{\leq i}$. To find the closest codeword to $\sum_{i \in \mathbb N_0}R_{i}(\xi)\eta^i$, the Viterbi algorithm considers at each step $i$ the restriction of the received word to the separation sets ${\cal S}_0$ to ${\cal S}_i$ (i.e., the word $\sum_{j=0}^i R_j(\xi)\eta^{j}$). Then, for each state $X_i(\xi)\in \mathbb F[\xi]^s_{\leq i}$, it takes all the paths from the zero state $X_0(\xi)=\bf{0} \in \mathbb F[\xi]_{\leq 0}$ to $X_i(\xi)$. For each of these paths with nodes $X_0(\xi), X_1(\xi), \dots, X_i(\xi)$, the algorithm considers the corresponding sequence of labels of the arcs $(U_0(\xi),W_0(\xi)), (U_1(\xi),W_1(\xi)), \dots, (U_i(\xi),W_i(\xi))$ and computes the minimal distance between $\sum_{j=0}^i W_j(\xi)\eta^{j}$ and $\sum_{j=0}^i R_j(\xi)\eta^{j}$. It sets $\bar d$ equal to the minimum of these distances and $\bar{sp}$ equal to the list of the paths with corresponding distance equal to $\bar d$ (the survivor paths). Then it assigns the label $(\bar d, \bar{sp})$ to $X_{i}(\xi)$.

This leads to the following Algorithm \ref{alg 2D viterbi}. A detailed example for this algorithm can be found in the appendix.

\begin{algorithm}[H]
    \DontPrintSemicolon
    \SetAlgoNoLine
    \SetAlgoNoEnd
    \caption{2D Viterbi algorithm}\label{alg 2D viterbi}
    \KwIn{ 
    A realization $\Sigma\!=\!\left(A_1,A_2,B_1,B_2,C,D\right)$ of dimension $s$ of an encoder $G(z_1,z_2)$ of $\cal C$ and the received word $R(\xi, \eta)=\sum_{i =0}^{N} R_i(\xi)\eta^{i}$, with $R_{N}(\xi) \neq 0$ where we call $N$ the global degree of $R(\xi,\eta)$.
    }
    \nlset{} \textbf{Initialization:} Set the initial node $X_0(\xi)=\mathbf{0}$ and assign to it the label $(d=0, sp= \emptyset)$. For all $i \in \mathbb N$, assign to all $X_i(\xi)\in\mathbb F[\xi]_{\leq i}^s$ the label $(\infty,\emptyset)$. Set $i=0$, $d_{min}=\infty$ and $sp_{min}= \emptyset$.
    \BlankLine
    \nl\label{1} \ForEach{ node $X_{i+1}(\xi)\in\mathbb F[\xi]_{\leq i+1}^s$}{
    \BlankLine
        \nlset{} \hspace{0.6cm }\nlset{1.1}\label{1.1} \hspace{-0.55cm}\ForEach{predecessor $X_i(\xi)\in\mathbb F[\xi]_{\leq i}^s$ with label $(d,sp)$}{compute the sum $d+d(W_i(\xi),R_i(\xi))$, where $(U_i(\xi),W_i(\xi))$ is the label of the arc $(X_i(\xi),X_{i+1}(\xi))$ in the trellis}
    \BlankLine
        \nlset{} \hspace{0.6cm }\nlset{1.2}\hspace{-0.55cm}Determine the minimum of the sums computed in \ref{1.1} and denote it by $\bar d$, and let $L$ be the list of the nodes $X_i(\xi)\in\mathbb F[\xi]_{\leq i}^s$ such that $\bar d=d+d(W_i(\xi),R_i(\xi))$, where $X_i(\xi)$ has label $(d,sp)$ and $(U_i(\xi),W_i(\xi))$ is the label of the arc $(X_i(\xi),X_{i+1}(\xi))$ in the state-transition diagram.\;
    \BlankLine
        \nlset{} \hspace{0.6cm }\nlset{1.3} \hspace{-0.55cm}Set $\bar{sp}=\emptyset$\\
    \BlankLine
        \nlset{} \hspace{0.6cm }\nlset{1.4} \hspace{-0.55cm}\ForEach{$X_i(\xi)$ in $L$ with label $(d,sp)$}{
        \uIf{$i=0$}{ set $\bar{sp}=\{\mathbf{0}\}$ }
        \uElse{set $\bar{sp}=\cup_{\ell\in sp} (\ell,X_i(\xi))\cup \bar{sp}$, where $(\ell,X_i(\xi))$ is the concatenation of the two vectors $\ell$ and $X_i(\xi)$}
        } 
    \BlankLine
        \nlset{} \hspace{0.6cm }\nlset{1.5} \hspace{-0.55cm}Assign to $X_{i+1}(\xi)$ the label $(\bar {d}, \bar{sp})$ ($\bar{sp}$ is the list of paths from $X_0(\xi)=0$ to $X_{i+1}(\xi)$, whose corresponding truncated codewords $\sum_{j\leq i } W_j(\xi)\eta^{j}$ have minimal distance to $\sum_{j\leq i} R_j(\xi)\eta^{j}$)\\
    \BlankLine
        \nlset{} \hspace{0.6cm }\nlset{1.6} \hspace{-0.55cm}\uIf{$i \geq N$ and $X_{i+1}(\xi)=\bf{0}$ and $\bar{d} = d_{min}$}{
        set $sp_{min}=sp_{min}\cup\bar{sp}$
        }
        \hspace{0.2cm}\uIf{$i \geq N$ and $X_{i+1}(\xi)=\bf{0}$ and $\bar{d} < d_{min}$}{set $d_{min}= \bar{d}$ and $sp_{min}=\bar{sp}$}
    } 
    \BlankLine
    \nl \uIf{all the nodes $X_{i+1}(\xi) \neq \bf{0}$ with label $(\bar {d}, \bar{sp})$ are such that $\bar{d}>d_{min}$}{
        \KwOut{the list $c_{min}$ of codewords corresponding to the paths in $sp_{min}$ and the corresponding messages}
        }
        \uElse{
        set $i=i+1$ and go to Step \ref{1}
        }

\end{algorithm}

The above algorithm returns $c_{min}$, which is the list of the codewords of $\cal C$ with minimal distance to the received word among all the codewords of the code. Since one cannot predict when the algorithm stops, it is usual to not consider all the codewords but only the ones with global degree at most $N$. The next algorithm is a modified version of Algorithm 2 that considers this restriction. The number of iterations will be $N$, from $i=0$ to $i=N-1$.

\begin{algorithm}[H]
    \DontPrintSemicolon
    \SetAlgoNoLine
    \SetAlgoNoEnd
    \caption{2D modified Viterbi algorithm}\label{alg mod 2D viterbi}
    \KwIn{ 
    A realization $\Sigma\!=\!\left(A_1,A_2,B_1,B_2,C,D\right)$ of dimension $s$ of an encoder $G(z_1,z_2)$ of $\cal C$ and the received word $R(\xi, \eta)=\sum_{i =0}^{N} R_i(\xi)\eta^{i}$, with $R_{N}(\xi) \neq 0$ where we call $N$ the global degree of $R(\xi,\eta)$.
    }
    \nlset{} \textbf{Initialization:} Set the initial node $X_0(\xi)=\mathbf{0}$ and assign to it the label $(d=0, sp= \emptyset)$. 
    For all $1\leq i \leq N$ assign to all $X_i(\xi)\in\mathbb F[\xi]_{\leq i}^s$ the label $(\infty,\emptyset)$.
        \BlankLine
    \nl \For{$i=0,\dots,N-1$}{\ForEach{$X_{i+1}(\xi)\in\mathbb F[\xi]_{\leq i+1}^s$}{
    \BlankLine
        \nlset{} \hspace{0.6cm }\nlset{1.1}\label{1.1 mod} \hspace{-0.55cm}\ForEach{predecessor $X_i(\xi)\in\mathbb F[\xi]_{\leq i}^s$ with label $(d,sp)$}{compute the sum $d+d(W_i(\xi),R_i(\xi))$, where $(U_i(\xi),W_i(\xi))$ is the label of the arc $(X_i(\xi),X_{i+1}(\xi))$ in the trellis}
    \BlankLine
        \nlset{} \hspace{0.6cm }\nlset{1.2}\hspace{-0.55cm}Determine the minimum of the sums computed in \ref{1.1 mod} and denote it by $\bar d$, and let $L$ be the list of the nodes $X_i(\xi)\in\mathbb F[\xi]_{\leq i}^s$ such that $\bar d=d+d(W_i(\xi),R_i(\xi))$, where $X_i(\xi)$ has label $(d,sp)$ and $(U_i(\xi),W_i(\xi))$ is the label of the arc $(X_i(\xi),X_{i+1}(\xi))$ in the state-transition diagram.\;
    \BlankLine
        \nlset{} \hspace{0.6cm }\nlset{1.3}\hspace{-0.55cm}Set $\bar{sp}=\emptyset$\\
    \BlankLine
        \nlset{} \hspace{0.6cm }\nlset{1.4}\hspace{-0.55cm}\ForEach{$X_i(\xi)$ in $L$ with label $(d,sp)$}{
        \uIf{$i=0$}{ set $\bar{sp}=\{\bf{0}\}$}
        \uElse{set $\bar{sp}=\cup_{\ell\in sp} (\ell,X_i(\xi))\cup \bar{sp}$, where $(\ell,X_i(\xi))$ is the concatenation of the two vectors $\ell$ and $X_i(\xi)$}
        } 
    \BlankLine
        \nlset{} \hspace{0.6cm }\nlset{1.5} \hspace{-0.55cm}Assign to $X_{i+1}(\xi)$ the label $(\bar {d}, \bar{sp})$ ($\bar{sp}$ is the list of paths from $X_0(\xi)=\bf{0}$ to $X_{i+1}(\xi)$, whose corresponding truncated codewords $\sum_{j\leq i } W_j(\xi)\eta^{j}$ have minimal distance to $\sum_{j\leq i} R_j(\xi)\eta^{j}$)\\
    } }
    \BlankLine
    \KwOut{the list $c_{min}$ of codewords corresponding to the paths in $\bar{sp}$ of the state $X_N(\xi)=\bf{0}$ and the corresponding messages.}
\end{algorithm}

\section{Complexity analysis and methods to obtain realizations of small dimension}

In this section we are going to analyze the complexity of the 2D modified Viterbi Algorithm; see Algorithm 3. As the following theorem shows this complexity is strongly related with the dimension of the linear system representation of the code that is used in the decoding algorithm.
Therefore, in this section, we also discuss the problem of finding realizations of smallest possible dimension, which we illustrate with several examples.

\begin{theorem}
For an $(n,k)$ 2D convolutional code and a received word of global degree $N$,
the complexity of Algorithm 3 is $$O(nN\cdot q^{(s+k)(N+1)-1}),$$
when using a system representation of dimension $s$.
\end{theorem}

\begin{proof}
   For each $i \in \{0,1,\dots,N-1\}$, Algorithm 3 considers the $q^{(i+1)s}$ nodes $X_{i}(\xi) \in \mathbb F[\xi]^s_{\leq i}$. For each of these nodes the algorithm considers $q^{(i+1)k}$ outgoing branches and for each of the predecessors it computes the distance between the $R_i(\xi) \in \mathbb F[\xi]^n_{\leq i}$ and the corresponding $W_i(\xi) \in \mathbb F[\xi]^n_{\leq i}$. Hence, the overall complexity is
   \begin{eqnarray*}
        O\left(\sum_{i=0}^{N-1} q^{(i+1)s} q^{(i+1)k} n(i+1) \right) & = & O\left(n N q^{s+k}\, \sum_{i=0}^{N-1} (q^{s+k})^i \right)
        \\
        & = & O\left(n N q^{s+k} \,\, \frac{q^{(s+k)N}-1}{q-1} \right).
   \end{eqnarray*}
  \end{proof}
Note that $N,n,k,q$ are fixed. However, $s$ depends on the dimension of the used system representation of the code. Therefore, a minimal system representation of the code, i.e., a system representation of minimal dimension among all the system representations of the code, would be desirable. Obtaining such a minimal system representation can be a difficult task. System representations of a code are realizations of encoders of the code. An encoder admits many realizations of different dimensions. The problem of obtaining a realization of minimal dimension among all realizations of an encoder is still an open problem in the 2D case.

A realization of an encoder of a 1D convolutional is minimal if and only if it is reachable and observable. The following definition considers such properties for 2D linear systems.

\begin{definition} \cite{FoMa86} \label{reach_observ2D}
Let $\Sigma=(A_1,A_2,B_1,B_2,C,D)$ be a 2D linear system with dimension
$s$. Let  
\begin{align*}
A_1 \, ^r\!\triangle^{t}  A_2 &=0, \text{ when either $r$ or $t$ is negative},\\
 A_1 \, ^r\!\triangle^{0}  A_2 &=A_1^r, \;\; A_1 \, ^0\!\triangle^{t}  A_2 =A_2^t,  \text{ for $r,t\ge0$},\\
 A_1 \, ^r\!\triangle^{t}   A_2& =A_1 \left(A_1\,^{r-1}\!\triangle^{t}  A_2\right)+ A_2\left( A_1\, ^r\!\triangle^{t-1}  A_2\right), \text{ for $r,t\ge1$}.
\end{align*}

\begin{enumerate}
 \item[\emph{(a)}]
     $\Sigma$ is called locally reachable if the reachability matrix
 \begin{align*}
  \mathcal{R}
    & = \left[
          R_1 \; \; R_2 \; \; R_3 \; \; \cdots
        \right] \mbox{ is full row rank},
\end{align*}
 where $R_k$ is defined as the block matrix including all columns given by
 $$\left(A_1 \, ^{i-1}\!\triangle^{j}   A_2\right)B_1 + \left(A_1 \, ^i\!\triangle^{j-1}   A_2\right)B_2$$ with $i+j=k$, for $i,j\ge0$ and

\item[\emph{(b)}] $\Sigma$ is called locally observable if the observability matrix
 \begin{align*}
  \mathcal{O}
    & = \left[\begin{array}{c}  \mathcal{O}_1 \\  \mathcal{O}_2 \\  \mathcal{O}_3 \\ \vdots
        \end{array}
        \right] \mbox{ is full row rank},
 \end{align*}
 where ${\cal O}_k$ is defined as the block matrix including all rows given by
 $$ C (A_1 \, ^{i}\!\triangle^{j}   A_2)$$ with $i+j=k$, for $i,j\ge0$.
\end{enumerate}
\end{definition}
However, in contrast to the 1D case, (local) reachability and (local) observability are not sufficient conditions for minimality of 2D linear systems. A characterization of minimal realizations of an encoder of a 2D convolutional code has been an open problem for many years.

In certain cases, a realization of lower dimension can be obtained from another realization satisfying certain conditions.
Given a realization $\Sigma=(A_1,A_2,B_1,B_2,C,D)$ of an encoder $G(z_1,z_2)$ of dimension $s$ and $T$ an invertible $s \times s$ matrix it follows that $$\tilde \Sigma=(T^{-1}A_1T,T^{-1}A_2T,T^{-1}B_1,T^{-1}B_2,CT,D)$$ is also a realization of $G(z_1,z_2)$ of the same dimension. $\Sigma$ and $\tilde \Sigma$ are said to be algebraically equivalent. In the particular case that
$\Sigma=(A_1,A_2,B_1,B_2,C,D)$ is the realization of $G(z_1,z_2)$ given by Algorithm \ref{Alg 1}, the corresponding reachability matrix has the identity matrix as a submatrix, which means that $\Sigma$ is locally reachable. Moreover, since $A_2A_1=0$, we have that its observability matrix is equal to
$$
\mathcal{O}=\left[ \begin{array}{c} CA_1^iA_2^j  \end{array} \right]_{i,j \in \mathbb N_0},
$$
up to row permutation. If $\rank\, \mathcal{O} = s$, then $\Sigma$ is also locally observable. If $\rank \, \mathcal{O} = r < s$, then $\dim \ker \mathcal{O}=s-r$. Let $(w_{r+1}, \dots, w_s)$ be a basis of $\ker \mathcal{O}$ and $T=\left[ \begin{array}{cccccc} w_1 & \cdots & w_r & w_{r+1}&  \dots & w_s  \end{array} \right] \in \mathbb F^{s \times s}$ be an invertible matrix, for some $w_1, \dots, w_r \in \mathbb F^s$. Note that $CT=\left[ \begin{array}{cc} C_1 & 0 \end{array} \right]$, where $C_1 \in \mathbb F^{n \times r}$. Let $x \in \ker \mathcal{O}$, i.e., 
\begin{equation}\label{kerO}
CA_1^iA_2^j x=0 \mbox{ for all } i,j \in \mathbb N_0.
\end{equation}
Then, $CA_1^iA_2^j (A_1 x)=0$ if $j \geq 1$ because $A_2A_1=0$, and also for $j=0$ by (\ref{kerO}). Thus, ${\cal O}(A_1x)=0$ and therefore, $\ker \cal{O}$ is $A_1$-invariant. On the other hand, (\ref{kerO}) implies that $CA_1^iA_2^{j+1}=0$ for all $i,j \in \mathbb N_0$ and therefore, ${\cal O}(A_2x)=0$ which means that $\ker \cal O$ is also $A_2$-invariant. Hence, 
$$
T^{-1}A_1T=\left[\begin{array}{cc} A_{11} & 0 \\ A_{13} & A_{14}\end{array}\right]
$$
and
$$
T^{-1}A_2T=\left[\begin{array}{cc} A_{21} & 0 \\ A_{23} & A_{24}\end{array}\right]
$$
where $A_{11},A_{21} \in \mathbb F^{r \times r}$, $A_{13},A_{23} \in \mathbb F^{(s-r) \times s}$ and $A_{14},A_{24} \in \mathbb F^{(s-r) \times (s-r)}$. Thus, we have that $$(T^{-1}A_1T, T^{-1}A_2T,T^{-1}B_1,T^{-1}B_2,CT,D)$$
is also a realization of $G(z_1,z_2)$ and therefore so is $(A_{11},A_{21},B_{11},B_{21},C_1,D)$, where $T^{-1}B_1=\left[\begin{array}{c} B_{11} \\ B_{12}\end{array}\right]$ and $T^{-1}B_2=\left[\begin{array}{c} B_{21} \\ B_{22}\end{array}\right]$, with $B_{11}, B_{21} \in \mathbb F^{r \times k}$ and $B_{12}, B_{22} \in \mathbb F^{(s-r) \times k}$ since 
\begin{eqnarray*}
    G(z_1,z_2) & = & \left[\begin{array}{cc} C_1 & 0 \end{array}\right]\left(\left[\begin{array}{cc} I_r & \\ & I_{s-r} \end{array}\right]-\left[\begin{array}{cc} A_{11} & 0 \\ A_{13} & A_{14}\end{array}\right]z_1-\left[\begin{array}{cc} A_{21} & 0 \\ A_{23} & A_{24}\end{array}\right]z_2\right)^{-1} \times \\
    & & \;\;\;\;\;
    \left(\left[\begin{array}{c} B_{11} \\ B_{12}\end{array}\right]z_1+\left[\begin{array}{c} B_{21} \\ B_{22}\end{array}\right]z_2\right) + D \\
    & = & C_1(I_r-A_{11} z_1 -A_{21} z_2)^{-1}\left(B_{11} z_1+B_{21} z_2\right) + D
\end{eqnarray*}
Moreover, $(A_{11},A_{21},B_{11},B_{21},C_1,D)$ is locally reachable since its reachability matrix is the submatrix of the matrix $T^{-1}{\cal R}$ consisting of the first $r$ rows and therefore it is full row rank. Furthermore, ${\cal O} T=\left[\begin{array}{cc} {\cal O}_1 & 0\end{array}\right]$, where ${\cal O}_1$ is the observability matrix of $(A_{11},A_{21},B_{11},B_{21},C_1,D)$. Since the rank of $\cal O$ is $r$, also the rank of ${\cal O}_1$ is $r$ and therefore, the system $(A_{11},A_{21},B_{11},B_{21},C_1,D)$ is locally observable.

Thus, given a realization of an encoder $G(z_1,z_2) \in \mathbb F^{n \times k}[z_1,z_2]$ obtained by Algorithm \ref{Alg 1} that is not locally observable, one can construct another realization of $G(z_1,z_2)$ with smaller dimension that is both locally observable and locally reachable. For clarity, this procedure is presented in the following algorithm.

\begin{algorithm}[H]
    \DontPrintSemicolon
    \caption{}\label{Alg 4}
    \KwIn{$\Sigma=(A_1,A_2,B_1,B_2,C,D)$  realization of dimension $s$ of $G(z_1,z_2)  \in \mathbb F[z_1,z_2]^{n \times k}$ given by Algorithm \ref{Alg 1} and $\cal O$ the observability matrix of $\Sigma$ of rank $r$ (with $r<s$).}
    \nl Let $T=\left[ \begin{array}{cccccc} w_1 & \cdots & w_r & w_{r+1}&  \dots & w_s  \end{array} \right] \in \mathbb F^{s \times s}$ be an invertible matrix such that $(w_{r+1}, \dots, w_s)$ is a basis of $\ker \mathcal{O}$.\;
    \nl Let $$ A_{11}:=\left[ \begin{array}{cc} I_r & 0 \end{array} \right]T^{-1}A_1 T\left[ \begin{array}{c} I_r \\ 0 \end{array}\right],$$ $$ A_{21}:=\left[ \begin{array}{cc} I_r & 0 \end{array} \right]T^{-1}A_2 T\left[ \begin{array}{c} I_r \\ 0 \end{array}\right],$$ $$ B_{11}:=\left[ \begin{array}{cc} I_r & 0 \end{array} \right]T^{-1}B_1,$$ $$ B_{21}:=\left[ \begin{array}{cc} I_r & 0 \end{array} \right]T^{-1}B_2, $$ $$ C_1:=C T\left[ \begin{array}{c} I_r \\ 0 \end{array}\right].$$\;
    \KwOut{$\Sigma\!=\!\left(A_{11},A_{21},B_{11},B_{21},C_1,D\right)$}
\end{algorithm}
We illustrate this way of reducing the dimension of our realization with the following examples, where $e_i$ denotes a standard basis column vector of suitable length.

\begin{example}\label{ex1}
  Consider the generator matrix  $G(z_1,z_2)=\begin{pmatrix}
        z_2z_1+z_2\\ z_1
    \end{pmatrix}\in\mathbb F_2^{2\times 1}[z_1,z_2]$.
        With Algorithm 1 we obtain the following realization of dimension $5$:
    \begin{align*}
        A_1&=\begin{pmatrix}
            0&0&0&0&0\\
            0&0&0&1&0\\
            0&0&0&0&1\\
            0&0&0&0&0\\
            0&0&0&0&0
        \end{pmatrix}, A_2 = \begin{pmatrix}
            0&0&0&1&0\\
            0&0&0&0&0\\
            0&0&0&0&0\\
            0&0&0&0&0\\
            0&0&0&0&0
        \end{pmatrix},
        C=\begin{pmatrix}
            0&1&0&1&0\\0&0&0&0&1
        \end{pmatrix},\\
       B_1&=\begin{pmatrix}
            0\\ 0\\ 0\\ 0\\ 1
        \end{pmatrix},
        B_2=\begin{pmatrix}
            0\\ 0\\ 0\\ 1\\ 0
        \end{pmatrix}
        , D=\begin{pmatrix}
            0\\ 0
        \end{pmatrix}.
        \end{align*}
Hence, 
the row-span of the corresponding observability matrix $\cal O$ is $\langle e_2^\top+e_4^\top,e_5^\top,e_4^\top\rangle$. Thus, $\cal O$ does not have full rank and therefore, this realization is not locally observable.

More specifically, the matrix $\mathcal{O}$ has rank $3$ and 
$\ker \mathcal{O}=\langle e_1,e_3 \rangle$. Considering the invertible matrix $$
T=\begin{pmatrix}
0&0&0&1&0\\1&0&0&0&0\\0&0&0&0&1\\0&1&0&0&0\\0&0&1&0&0
        \end{pmatrix},$$
we obtain $$ A_{11}=\left[ \begin{array}{cc} I_3 & 0 \end{array} \right]T^{-1}A_1 T\left[ \begin{array}{c} I_3 \\ 0 \end{array}\right]=\begin{pmatrix}
0&1&0\\0&0&0\\0&0&0
        \end{pmatrix},$$ $$ A_{21}=\left[ \begin{array}{cc} I_3 & 0 \end{array} \right]T^{-1}A_2 T\left[ \begin{array}{c} I_3 \\ 0 \end{array}\right]=\begin{pmatrix}
0&0&0\\0&0&0\\0&0&0
        \end{pmatrix},$$ $$ B_{11}=\left[ \begin{array}{cc} I_3 & 0 \end{array} \right]T^{-1}B_1=\begin{pmatrix}
0\\0\\1
        \end{pmatrix},$$ $$ B_{21}=\left[ \begin{array}{cc} I_3 & 0 \end{array} \right]T^{-1}B_2=\begin{pmatrix}
0\\1\\0  \end{pmatrix}, $$ $$ C_1=C T\left[ \begin{array}{c} I_3 \\ 0 \end{array}\right]=\begin{pmatrix}
1&1&0\\0&0&1
        \end{pmatrix}.$$
$\Sigma\!=\!\left(A_{11},A_{21},B_{11},B_{21},C_1,D\right)$ is a locally reachable and locally observable realization of $G(z_1,z_2)$ of dimension $3$.
\end{example}

\begin{example}
    For $G(z_1,z_2)=\begin{pmatrix}
        z_2^2z_1\\ z_2z_1+z_2
    \end{pmatrix}\in\mathbb F_2^{2\times 1}[z_1,z_2]$, we obtain with Algorithm \ref{Alg 1} the $9$ dimensional realization $\Sigma=(A_1,A_2,B_1,B_2,C,D)$ with 
    \begin{align*}
        A_1&=\begin{pmatrix}
            0&0&0&0&0&0&0&0&0\\
            0&0&0&0&1&0&0&0&0\\
            0&0&0&0&0&1&0&0&0\\
            0&0&0&0&0&0&1&0&0\\
            0&0&0&0&0&0&0&0&0\\
            0&0&0&0&0&0&0&1&0\\
            0&0&0&0&0&0&0&0&1\\
            0&0&0&0&0&0&0&0&0\\
            0&0&0&0&0&0&0&0&0\\
        \end{pmatrix},
        A_2=\begin{pmatrix}
            0&0&0&0&1&0&0&0&0\\
            0&0&0&0&0&0&0&0&0\\
            0&0&0&0&0&0&0&0&0\\
            0&0&0&0&0&0&0&0&0\\
            0&0&0&0&0&0&0&1&0\\
            0&0&0&0&0&0&0&0&0\\
            0&0&0&0&0&0&0&0&0\\
            0&0&0&0&0&0&0&0&0\\
            0&0&0&0&0&0&0&0&0\\
        \end{pmatrix},\\
        B_1&=\begin{pmatrix}
            0&0&0&0&0&0&0&0&1
        \end{pmatrix}^T,
        B_2=\begin{pmatrix}
            0&0&0&0&0&0&0&1&0
        \end{pmatrix}^T,\\
        C&=\begin{pmatrix}
            0&1&0&0&0&0&0&0&0\\
            0&0&0&0&0&1&0&1&0
        \end{pmatrix},
        D=\begin{pmatrix}
            0\\0
        \end{pmatrix}
    \end{align*}
    and $\Psi(z_1,z_2)=\begin{pmatrix}
            z_2^3 & z_2^2z_1 & z_2z_1^2 & z_1^3 & z_2^2 & z_2z_1 & z_1^2 & z_2 & z_1
        \end{pmatrix}^T$.

The row-span of $\cal O$ is $\langle e_2^\top,e_6^\top+e_8^\top,e_5^\top,e_8^\top \rangle$  and therefore $\ker \mathcal{O}=\langle e_1,e_3,e_4,e_7,e_9 \rangle$. Considering the invertible matrix $$
T=\begin{pmatrix}
0&0&0&0&1&0&0&0&0\\1&0&0&0&0&0&0&0&0\\0&0&0&0&0&1&0&0&0\\0&0&0&0&0&0&1&0&0\\0&1&0&0&0&0&0&0&0\\0&0&1&0&0&0&0&0&0 \\ 0&0&0&0&0&0&0&1&0\\0&0&0&1&0&0&0&0&0\\0&0&0&0&0&0&0&0&1
        \end{pmatrix},$$
we obtain the locally reachable and locally observable realization of $G(z_1,z_2)$ of dimension $4$, $\Sigma\!=\!\left(A_{11},A_{21},B_{11},B_{21},C_1,D\right)$, where $$ A_{11}=\left[ \begin{array}{cc} I_4 & 0 \end{array} \right]T^{-1}A_1 T\left[ \begin{array}{c} I_4 \\ 0 \end{array}\right]=\begin{pmatrix}
0&1&0&0\\0&0&0&0\\0&0&0&1\\0&0&0&0
        \end{pmatrix},$$ $$ A_{21}=\left[ \begin{array}{cc} I_4 & 0 \end{array} \right]T^{-1}A_2 T\left[ \begin{array}{c} I_4 \\ 0 \end{array}\right]=\begin{pmatrix}
0&0&0&0\\0&0&0&1\\0&0&0&0\\0&0&0&0
        \end{pmatrix},$$ $$ B_{11}=\left[ \begin{array}{cc} I_4 & 0 \end{array} \right]T^{-1}B_1=\begin{pmatrix}
0\\0\\0\\0
        \end{pmatrix},$$ $$ B_{21}=\left[ \begin{array}{cc} I_4 & 0 \end{array} \right]T^{-1}B_2=\begin{pmatrix}
0\\0\\0\\1  \end{pmatrix}, $$ $$ C_1=C T\left[ \begin{array}{c} I_4 \\ 0 \end{array}\right]=\begin{pmatrix}
1&0&0&0\\0&0&1&1
        \end{pmatrix}.$$

In the following, we will explain how to obtain a realization with only dimension $3$. The only monomials appearing in $G(z_1,z_2)$ are $z_2^2z_1$, $z_2z_1$ and $z_2$. In this particular example, it is sufficient to choose
$\hat \Psi(z_1,z_2)=\begin{pmatrix}
        z_2^2z_1 & z_2z_1 & z_2
    \end{pmatrix}^T$ to be able to construct the realization 
    $\Sigma=(\hat A_1,\hat A_2,\hat B_1, \hat B_2,\hat C,D)$ with
    \begin{align*}
        \hat A_1=\begin{pmatrix}
            0&0&0\\
            0&0&1\\
            0&0&0
        \end{pmatrix},
        \hat A_2=\begin{pmatrix}
            0&1&0\\
            0&0&0\\
            0&0&0
        \end{pmatrix},
        \hat B_1=\begin{pmatrix}
            0\\0\\0
        \end{pmatrix},
        \hat B_2=\begin{pmatrix}
            0\\0\\1
        \end{pmatrix},
        \hat C=\begin{pmatrix}
            1&0&0\\
            0&1&1
        \end{pmatrix}, 
        D=\bf{0}.
    \end{align*}
    This realization is of smaller dimension $3$ and moreover, it is locally observable since 
    \begin{align*}
        \mathcal{O}=\begin{pmatrix}
            1&0&0\\
            0&1&1\\
            \hline
            0&0&0\\
            0&0&1\\
            \hline
            \vdots&\vdots&\vdots
        \end{pmatrix}
    \end{align*}
    has full rank.
    In general, it does not always work to just choose the monomials appearing in $G(z_1,z_2)$ to form $\hat{\Psi}(z_1,z_2)$ but we might need to add some further monomials to this set.
In particular, if $z_1^az_2^b$ for $a+b\geq 2$ is included in this set, then we also need to have either $z_1^{a-1}z_2^b$ or $z_1^az_2^{b-1}$ included in this set.
In this example, this condition is automatically fulfilled by just choosing the monomials appearing in $G(z_1,z_2)$ and we do not need to add further monomials. However, later in the paper, we illustrate with Example \ref{ex:2cols} a case where additional monomials besides the ones appearing in $G(z_1,z_2)$ need to be added to the set $\hat{\Psi}(z_1,z_2)$.    
An algorithm on how to form this set in general and how obtain a realization from it is presented in the following Algorithm \ref{alg more improved}.
\end{example}

\begin{algorithm}[H]
    \DontPrintSemicolon
    \SetAlgoNoEnd
    \SetAlgoNoLine
    \caption{}\label{alg more improved}
    \KwIn{$G(z_1,z_2) \in \mathbb F[z_1,z_2]^{n \times k}$ with column degrees $\nu_1,\nu_2,\dots, \nu_k$}
    \nl Let $D:=G(0,0)$ and $\hat G(z_1,z_2):=G(z_1,z_2)-G(0,0)$.\;
    \nl \For{$i=1,\dots,k$}{Let $\hat S_i$ be a minimal set of monomials such that
    \begin{itemize}
        \item all monomials that appear in column $i$ are in $\hat S_i$
        \item for $z_1^az_2^b\in \hat S_i$ with $a+b\geq 2$ follows $$a>0 \text{ and }z_1^{a-1}z_2^{b}\in \hat S_i \qquad \text{ or }\qquad b>0 \text{ and }z_1^{a}z_2^{b-1}\in \hat S_i.$$
    \end{itemize}
    With $\mathcal{Z}^{(\nu_i)}$ from Algorithm \ref{Alg 1} define $\hat{\mathcal{I}}_i:=\{j=1,\dots,\frac{\nu_i(\nu_i+3)}{2}\mid \mathcal{Z}^{(\nu_i)}_j\in \hat S_i\}$ 
    and let $\hat{\mathcal{Z}}^{(i)}:= \mathcal{Z}^{(\nu_i)}\mid_{\hat{\mathcal{I}}_i}$ be the vector restricted to the entries with indices in $\hat{\mathcal{I}}_i$.}
    \nl Define $$ \hat{\Psi}(z_1,z_2) ={\rm diag}\{\hat{\mathcal{Z}}^{(1)},\hat{\mathcal{Z}}^{(2)},\dots,\hat{\mathcal{Z}}^{(k)}\} $$ and let $\hat{C} \in \mathbb F^{n \times \vert \hat{\mathcal{I}}_i\vert}$ be such that $\hat G(z_1,z_2)=\hat{C}\hat{\Psi}(z_1,z_2)$.\;
        \nl \For{$i=1,\dots,k$}{Let $\hat A_1^{(i)}= \hat A_2^{(i)}=\mathbf{0}\in\mathbb{F}_q^{\vert\hat{\mathcal{I}}_i\vert\times \vert\hat{\mathcal{I}}_i\vert}$ and $\hat B_1^{(i)}=\hat B_2^{(i)}=\mathbf{0}\in\mathbb{F}_q^{\vert\hat{\mathcal{I}}_i\vert}$.
    \;
    \For{$\ell=1,\dots,\vert\hat{\mathcal{I}}_i\vert$}{ 
    Let $a,b\in\mathbb{N}_0$ be such that $\hat{\mathcal{Z}}^{(i)}_\ell=z_1^az_2^b$.\;
        \If{$a=1$ and $b=0$}{write a $1$ in row $\ell$ of $\hat B_1^{(i)}$}
        \If{$a=0$ and $b=1$}{write a $1$ in row $\ell$ of $\hat B_2^{(i)}$}
        \If{$z_1^{a-1}z_2^b\in\hat S_i$}{ there exists some $\hat\ell$ such that $\hat{\mathcal{Z}}^{(i)}_{\hat \ell}=z_1^{a-1}z_2^b$\;write a $1$ at position $(\ell,\hat\ell)$ in $\hat A_1^{(i)}$}
        \Else{there exists some $\hat\ell$ such that $\hat{\mathcal{Z}}^{(i)}_{\hat \ell}=z_1^{a}z_2^{b-1}$\;write a $1$ at position $(\ell,\hat\ell)$ in $\hat A_2^{(i)}$}
    }}
    \nl Let $$ \hat{A_1}:={\rm diag}\{\hat{A_1}^{(1)},\hat{A_1}^{(2)},\dots,\hat{A_1}^{(k)}\},$$ $$ \hat{A_2}:={\rm diag}\{\hat{A_2}^{(1)},\hat{A_2}^{(2)},\dots,\hat{A_2}^{(k)}\}, $$ $$\hat{B_1}:={\rm diag}\{\hat{B_1}^{(1)},\hat{B_1}^{(2)},\dots,\hat{B_1}^{(k)}\},$$ $$ \hat{B_2}:={\rm diag}\{\hat{B_2}^{(1)},\hat{B_2}^{(2)},\dots,\hat{B_2}^{(k)}\}.$$
    In case $\nu_i=0$ for some $i$, the $i$-th columns in $\hat B_{1}$, $\hat B_{2}$ and $\hat{\Psi}(z_1,z_2)$ are equal to zero, and the $i$-th diagonal blocks $\hat A_{1}^{(i)}$ and $\hat A_{2}^{(i)}$ are empty.\;
    \KwOut{$\Sigma\!=\!\left(\hat{A_1},\hat{A_2},\hat{B_1},\hat{B_2},\hat{C}, D\right)$}
\end{algorithm}

\begin{theorem}
    Algorithm \ref{alg more improved} gives a locally reachable realization of the encoder $G(z_1,z_2) \in \mathbb F[z_1,z_2]^{n \times k}$. 
\end{theorem}
\begin{proof}
    We need to show that for each column $i=1,\dots,k$ of $G(z_1,z_2)$, the equation
    \begin{equation}\label{equ}
        (I-\hat A_1^{(i)}z_1-\hat A_2^{(i)}z_2)\hat{\mathcal{Z}}^{(i)}=\hat B_1^{(i)}z_1+\hat B_2^{(i)}z_2
    \end{equation}
    holds.
    For each $j=1,\dots,\vert\hat{\mathcal{I}}_i\vert$, we have exactly one of the following $4$ cases:
    
    \textbf{Case 1:} $\hat{\mathcal{Z}}_j^{(i)}=z_1$\\
        Then, the $j$-th entry on the left side of equation (\ref{equ}) is equal to $z_1$, and the $j$-th entry on the right side is also equal to $z_1$.

    \textbf{Case 2:} $\hat{\mathcal{Z}}_j^{(i)}=z_2$\\
        Then, the $j$-th entry on the left side of equation (\ref{equ}) is equal to $z_2$, and the $j$-th entry on the right side is also equal to $z_2$.
        
    \textbf{Case 3:} $\hat{\mathcal{Z}}_j^{(i)}=z_1^az_2^b$ for $a+b\geq 2$ and there exists some $\hat j$ such that $\hat{\mathcal{Z}}_{\hat{j}}^{(i)}=z_1^{a-1}z_2^b$\\
        Then, the $j$-th entry on the left side of equation (\ref{equ}) is equal to $\hat{\mathcal{Z}}_j^{(i)}-z_1\hat{\mathcal{Z}}_{\hat{j}}^{(i)}=z_1^az_2^b-z_1z_1^{a-1}z_2^b=0$, and the $j$-th entry on the right side is also equal to $0$.
    
    \textbf{Case 4:} $\hat{\mathcal{Z}}_j^{(i)}=z_1^az_2^b$ for $a+b\geq 2$ and there exists some $\hat j$ such that $\hat{\mathcal{Z}}_{\hat{j}}^{(i)}=z_1^{a}z_2^{b-1}$ and $z_1^{a-1}z_2^b\notin \hat S_i$\\
        Then, the $j$-th entry on the left side of equation (\ref{equ}) is equal to $\hat{\mathcal{Z}}_j^{(i)}-z_2\hat{\mathcal{Z}}_{\hat{j}}^{(i)}=z_1^az_2^b-z_2z_1^{a-1}z_2^b=0$, and the $j$-th entry on the right side is also equal to $0$.

    Therefore, equation (\ref{equ}) holds, and Algorithm \ref{alg more improved} gives a realization of $G(z_1,z_2)$.

   It remains to show that each realization obtained form Algorithm \ref{alg more improved} is locally reachable.
    Because of the block diagonal structure of the matrices $\hat{A}_1,\hat{A_2}, \hat{B}_1,\hat{B}_2$ it is enough to show this statement for $k=1$. We proof it by showing that each non-zero column of $\hat{A}_1$ and $\hat{A}_2$ is also a column of $\mathcal{R}$, which implies that $\mathcal{R}$ contains the identity matrix.\\
    Each row of the matrix $\begin{pmatrix}
        \hat{A}_1 & \hat{A}_2 & \hat{B}_1 & \hat{B_2}
    \end{pmatrix}$ contains exactly one non-zero entry and this entry is equal to $1$. Moreover, each column of $\begin{pmatrix}
        \hat{A}_1 & \hat{A}_2 & \hat{B}_1 & \hat{B_2}
    \end{pmatrix}$ contains at most one non-zero entry. Note that these last two properties already imply that the matrix $\begin{pmatrix}
        \hat{A}_1 & \hat{A}_2 & \hat{B}_1 & \hat{B_2}
    \end{pmatrix}$ contains the identity matrix.\\
    Consider an arbitrary non-zero column of $\hat{A}_1$. The row of its non-zero entry corresponds to a polynomial $z_1^az_2^b$ from $\hat{\Psi}(z_1,z_2)$ for which also $z_1^{a-1}z_2^b$ is in $\hat{\Psi}(z_1,z_2)$. Overall, there is a chain of monomials from $\hat{\Psi}(z_1,z_2)$ going from $z_1^az_2^b$ to either $z_1$ or $z_2$ where in each step of the chain the exponent of either $z_1$ or $z_2$ is decreased by $1$ and if both is possible in $\hat{\Psi}(z_1,z_2)$, then the exponent of $z_1$ is decreased. Each time the exponent of $z_1$ is decreased this corresponds to a $1$ in $\hat{A}_1$ and each time the exponent of $z_2$ is decreased this corresponds to a $1$ in $\hat{A}_2$. Overall, if we multiply all the matrices corresponding to such a chain with $\hat{B}_1$ if we end with $z_1$ and with $\hat{B}_2$ if we end with $z_2$, we obtain the column of $\hat{A}_1$ we started with. For example if $\hat{\Psi}(z_1,z_2)=\{z_2^2z_1,z_2^3,z_2^2, z_2z_1,z_2\}$,
    i.e., 
    \begin{align*}
        \hat{A}_1=\begin{pmatrix}
            0 & 0 & 1 & 0 & 0\\
            0 & 0 & 0 & 0 & 0\\
            0 & 0 & 0 & 0 & 0\\
            0 & 0 & 0 & 0& 1\\
            0 & 0 & 0 & 0 & 0
        \end{pmatrix}, \hat{A}_2=\begin{pmatrix}
            0 & 0 & 0 & 0 & 0\\
            0 & 0 & 1 & 0 & 0\\
            0 & 0 & 0 & 0 & 1\\
            0 & 0 & 0 & 0 & 0\\
            0 & 0 & 0 & 0 & 0
        \end{pmatrix}, \hat{B}_1=\mathbf{0}, \hat{B}_2=e_5,
    \end{align*}
    then the chain $z_2^2z_1\rightarrow z_2^2\rightarrow z_2$ corresponds to the fact that $\hat{A}_1\hat{A}_2\hat{B}_2$ is equal to the third column of $\hat{A}_1$. Overall, all columns of $\hat{A}_1$ are equal to a summand of $\left(\hat{A}_1 \, ^{i-1}\!\triangle^{j}   \hat{A}_2\right)\hat{B}_1 + \left(\hat{A}_1 \, ^i\!\triangle^{j-1}   \hat{A}_2\right)\hat{B}_2$ for some for $i,j\geq 0$, in the above example, $i=1$ and $j=2$.
    It remains to show that all the other summands of such a sum are equal to zero. For fixed $i$ and $j$, all summands correspond to possible chains of monomials with the same starting and end point either $z_1$ or $z_2$. Since Algorithm \ref{alg more improved} places $1$s in $\begin{pmatrix}
        \hat{A}_1 & \hat{A}_2 & \hat{B}_1 & \hat{B_2}
    \end{pmatrix}$ for only one of these paths, all other summands are zero. For the above example this means that $\hat{A}_2\hat{A}_1\hat{B}_2=\bf{0}$ (corresponding to the chain $z_2^2z_1\rightarrow z_1z_2\rightarrow z_2$) and $\hat{A}_2^2\hat{B}_1=\bf{0}$
    (corresponding to the chain $z_2^2z_1\rightarrow z_1z_2\rightarrow z_1$).
    \\
    This shows that all non-zero columns of $\hat{A}_1$ are also columns of $\mathcal{R}$. With exactly the same reasoning one can show that all non-zero columns of $\hat{A}_2$ are also columns of $\mathcal{R}$, which finishes the proof.
\end{proof}

If we apply the previous algorithm to the Example \ref{ex1}, we get the same realization as with Algorithm \ref{Alg 4}.

  \begin{example} As in Example \ref{ex1} we consider 
 $G(z_1,z_2)=\begin{pmatrix}
        z_2z_1+z_2\\ z_1
    \end{pmatrix}$.
      With Algorithm \ref{alg more improved} we obtain:
$$\hat \Psi(z_1,z_2)=\begin{pmatrix}
        z_2z_1 & z_2 & z_1
    \end{pmatrix}^T, 
    \hat{B}_2=\begin{pmatrix}
        0\\1\\0
    \end{pmatrix},
    \hat{B}_1=\begin{pmatrix}
        0\\0\\1
    \end{pmatrix}$$
     
    \begin{align*}
        \hat{A}_1=\begin{pmatrix}
            0&1&0\\0&0&0\\0&0&0
        \end{pmatrix}, \hat{A}_2=\begin{pmatrix}
            0&0&0\\0&0&0\\0&0&0
        \end{pmatrix}, \hat{C}=\begin{pmatrix}
            1&1&0\\0&0&1
        \end{pmatrix}.
    \end{align*}
    \end{example}

In the following, we present an example where Algorithm \ref{Alg 4} and Algorithm \ref{alg more improved} give again realizations of the same dimension but now the realizations itself are different. Actually, in this example, the realization from Algorithm \ref{alg more improved} is not locally observable.

\begin{example}
 Consider $G(z_1,z_2)=\begin{pmatrix}
        z_2^2z_1^2+z_2+z_1\\ z_2^2z_1+z_2z_1
    \end{pmatrix}\in\mathbb F_2^{2\times 1}[z_1,z_2]$.    
Applying Algorithm \ref{Alg 1}, we obtain the realization $(A_1,A_2,B_1,B_2,C,D)$ where
$$
A_1=\left[\begin{array}{cccccccccccccc}      0&0&0&0&0&e_2&e_3&e_4&e_5&e_7&e_8&e_9&e_{11}&e_{12}
    \end{array}\right],
$$
$$
A_2=\left[\begin{array}{cccccccccccccc}      0&0&0&0&0&e_1&0&0&0&e_6&0&0&e_{10}&0
    \end{array}\right],
$$
$$
B_1=\left[\begin{array}{c}  e_{14}
    \end{array}\right], \;\;\; B_2=\left[\begin{array}{c}  e_{13}
    \end{array}\right], $$
    $$ C=\left[\begin{array}{cc} e_3+e_{13}+e_{14} & e_7+e_{11}
    \end{array}\right]^T, \;\;\; D=\left[\begin{array}{c} 0 \\ 0
    \end{array}\right].
$$ 
The row-span of the observability matrix $\cal O$ is $\langle e_3^T+e_{14}^T,e_7^T, e_{10}^T,e_{11}^T,e_{13}^T \rangle$ and therefore, $\cal O$ has rank $5$ and $\ker \mathcal{O}=\langle e_1,e_2,-e_3+e_{14},e_4,e_5,e_6,e_8,e_9,e_{12} \rangle$.
Applying now Algorithm \ref{Alg 4}, we consider
$$
T=\left[\begin{array}{cccccccccccccc}      e_3&e_7&e_{10}&e_{11}&e_{13}&e_1&e_2&-e_3+e_{14}&e_4&e_5&e_6&e_8&e_9&e_{12}
    \end{array}\right]
$$
to obtain the locally reachable and locally observable realization of $G(z_1,z_2)$ of dimension $5$, $\Sigma\!=\!\left(A_{11},A_{21},B_{11},B_{21},C_1,D\right)$, where $$ A_{11}=\left[ \begin{array}{cc} I_5 & 0 \end{array} \right]T^{-1}A_1 T\left[ \begin{array}{c} I_5 \\ 0 \end{array}\right]=\begin{pmatrix}
0&1&0&0&0\\0&0&1&0&0\\0&0&0&0&0\\0&0&0&0&1\\0&0&0&0&0 \end{pmatrix},$$ $$ A_{21}=\left[ \begin{array}{cc} I_5 & 0 \end{array} \right]T^{-1}A_2 T\left[ \begin{array}{c} I_5 \\ 0 \end{array}\right]=\begin{pmatrix}
0&0&0&0&0\\0&0&0&0&0\\0&0&0&0&1\\0&0&0&0&0\\0&0&0&0&0 \end{pmatrix},$$ $$ B_{11}=\left[ \begin{array}{cc} I_5 & 0 \end{array} \right]T^{-1}B_1=\begin{pmatrix}
1\\0\\0\\0\\0
        \end{pmatrix},$$ $$ B_{21}=\left[ \begin{array}{cc} I_5 & 0 \end{array} \right]T^{-1}B_2=\begin{pmatrix}
0\\0\\0\\0\\1  \end{pmatrix}, $$ $$ C_1=C T\left[ \begin{array}{c} I_5 \\ 0 \end{array}\right]=\begin{pmatrix}
1&0&0&0&1\\0&1&0&1&0
        \end{pmatrix}.$$

In contrast, with Algorithm \ref{alg more improved}, we obtain the following realization:
\begin{align*}
\hat \Psi(z_1,z_2)&=\begin{pmatrix}
      z_2^2z_1^2 & z_2^2z_1 &  z_2z_1 & z_2 & z_1
    \end{pmatrix}^T,\\
    \hat{A}_1&=\begin{pmatrix}
         0 & 1 & 0 & 0 & 0\\
         0 & 0 & 0 & 0 & 0\\
          0 & 0 & 0 & 1 & 0\\
           0 & 0 & 0 & 0 & 0\\
            0 & 0 & 0 & 0 & 0
    \end{pmatrix},
    \hat{A}_2=\begin{pmatrix}
         0 & 0 & 0 & 0 & 0\\
         0 & 0 & 1 & 0 & 0\\
          0 & 0 & 0 & 0 & 0\\
           0 & 0 & 0 & 0 & 0\\
            0 & 0 & 0 & 0 & 0
    \end{pmatrix},\\
     \hat{B}_1&=\begin{pmatrix}
      0\\ 0\\  0\\0\\1
    \end{pmatrix},
    \hat{B}_2=\begin{pmatrix}
      0\\ 0\\  0\\1\\0
    \end{pmatrix},
    \hat{C}=\begin{pmatrix}
        1  & 0 & 0 & 1 & 1\\
        0 & 1 & 1 & 0 & 0
    \end{pmatrix},   
\end{align*}
i.e., the row-span of $\hat{\mathcal{O}}$ is $\langle e_1^\top+e_5^\top,e_2^\top,e_3^\top,e_4^\top\rangle$ and hence, $\hat{\mathcal{O}}$  
is not full rank and the realization is not locally observable.

In the following, we explain a method with which we can further decrease the dimension of the realization. To this end we consider again equation \eqref{equ} and search for matrices $A_1$, $A_2$, $B_1$, $B_2$, $\mathcal{Z}$ of smallest possible size fulfilling this equation such that the entries of $\mathcal{Z}$ form a set of polynomials from $\mathbb F[z_1,z_2]$ whose $\mathbb F$-span includes the set of entries of $G(z_1,z_2)$ (and $A_1$, $A_2$, $B_1$, $B_2$ have entries from $\mathbb F$). For our concrete example the following equation provides such a solution of dimension $4$:
$$\begin{pmatrix}
    1 & -z_1 & 0 & 0\\
    0 & 1 & -z_2 & 0\\
    0 & 0 & 1 & -z_2\\
    0 & 0 & 0 & 1
\end{pmatrix}\begin{pmatrix}
    z_2^2z_1^2+z_2\\
    z_2^2z_1\\ z_1z_2\\ z_1
\end{pmatrix}=\begin{pmatrix}
    z_2\\ 0\\ 0\\ z_1
\end{pmatrix},$$
i.e.
\begin{align*}
A_1=\begin{pmatrix}
    0 & 1 & 0 & 0\\
    0 & 0 & 0 & 0\\
      0 & 0 & 0 & 0\\
      0 & 0 & 0 & 0
\end{pmatrix}\hspace{-1mm}, A_2=\begin{pmatrix}
    0 & 0 & 0 & 0\\
    0 & 0 & 1 & 0\\
      0 & 0 & 0 & 1\\
      0 & 0 & 0 & 0
\end{pmatrix}\hspace{-1mm},
B_1=\begin{pmatrix}
    0\\0\\0\\1
\end{pmatrix}\hspace{-1mm}, B_2=\begin{pmatrix}
    1\\0\\0\\0
\end{pmatrix}\hspace{-1mm}, \mathcal{Z}=\begin{pmatrix}
    z_2^2z_1^2+z_2\\
    z_2^2z_1\\ z_1z_2\\ z_1
\end{pmatrix}.
\end{align*}
Then, we solve the equation $\hat G(z_1,z_2)=C\mathcal{Z}$ (in our case $D=G(0,0)=0$, i.e., $\hat G(z_1,z_2)=G(z_1,z_2)$) for $C$ to obtain
$$C=\begin{pmatrix}
    1 & 0 & 0 & 1\\
    0 & 1 & 1 & 0
\end{pmatrix}.$$
Calculating 
$$\mathcal{O}=\begin{pmatrix}
    1 & 0 & 0 & 1\\
    0 & 1 & 1 & 0\\
     0 & 1 & 0 & 0\\
       0 & 0 & 0 & 0\\
      0 & 0 & 0 & 0\\
      0 & 0 & 1 & 1\\
      \vdots
\end{pmatrix},$$
which is full rank, shows that this realization also is locally observable.
\end{example}

In the following, we will formalize the idea presented at the end of the preceding example.

\begin{lemma}
Let $s_j$ be the minimal number of elements of $\mathbb F[z_1,z_2]$ needed such that each entry of column $j$ of $G(z_1,z_2)\in\mathbb F_q[z_1,z_2]^{k\times n}$ is in the $\mathbb F_q$-span of these elements. Then, the dimension of a minimal realization of $G(z_1,z_2)$ is lower bounded by $\max_{j\in\{1,\hdots,k\}}\{s_j\}$.
\end{lemma}
\begin{proof}
Without loss of generality assume $D=G(0,0)=0$. We have that 
$G(z_1,z_2)=C(I-A_1z_1-A_2z_2)^{-1}(B_1z_1+B_2z_2)$. Each entry in column $j$ of $G(z_1,z_2)$ is a linear combination (with constant coefficients that are entries of $C$) of the entries of the $j$-th column of $(I-A_1z_1-A_2z_2)^{-1}(B_1z_1+B_2z_2)$. 
\end{proof}

The method used in the end of the previous example to further reduce the dimension is exactly the search for the minimal number of elements whose $\mathbb F_2$-span covers all entries of $G(z_1,z_2)$ such that additionally all of these elements can be obtained as entries of $(I-A_1z_1-A_2z_2)^{-1}(B_1z_1+B_2z_2)$ for suitable matrices $A_1,A_2,B_1,B_2$.
There is no general algorithm telling us how to find such a minimal realization but this has to be considered case by case for each specific example.

Besides reducing the dimension of the previous example in this way, this allows us to also further reduce the dimension of Example \ref{ex1}, even if Algorithm \ref{Alg 4} and Algorithm \ref{alg more improved} gave the same observable realization.

\begin{example}
Considering again  $G(z_1,z_2)=\begin{pmatrix}
        z_2z_1+z_2\\ z_1
    \end{pmatrix}$,   
we observe that using
$$\begin{pmatrix}
    1 & -z_2\\
    0 & 1
\end{pmatrix}\begin{pmatrix}
    z_1z_2+z_2\\ z_1
\end{pmatrix}=\begin{pmatrix}
    z_2\\ z_1
\end{pmatrix}$$
leads to the $2$-dimensional realization 
$$A_1=\begin{pmatrix}
    0 & 0\\ 0 & 0
\end{pmatrix}, 
A_2=\begin{pmatrix}
    0 & 1\\ 0 & 0
\end{pmatrix}, 
B_1=\begin{pmatrix}
    0\\ 1
\end{pmatrix}, 
B_2=\begin{pmatrix}
    1\\ 0
\end{pmatrix}, C=\begin{pmatrix}
    1 & 0\\ 0 & 1
\end{pmatrix},$$
which clearly is observable.
\end{example}

We conclude this section with the following example that shows that for generator matrices with $k>1$, the choice of the generator matrix matters for the size of the realization we obtain.

\begin{example}\label{ex:2cols}
    Consider 
    $$G(z_1,z_2)=\begin{pmatrix}
        z_1^2 & z_1^2\\
        z_1z_2& z_2\\
        z_2   & 0
    \end{pmatrix}\quad \text{and}\quad \Tilde{G}(z_1,z_2)=\begin{pmatrix}
        0 & z_1^2\\
        z_2+z_1z_2&z_2\\
        z_2 &0
    \end{pmatrix},$$ which are generator matrices of the same convolutional code. Applying Algorithm \ref{alg more improved} to the matrix $G(z_1,z_2)$ we obtain a realization of dimension $7$, in particular
    $$\hat{\Psi}(z_1,z_2)=\begin{pmatrix}
        z_1z_2&z_1^2&z_2&z_1\\
        &&&&z_1^2&z_2&z_1
    \end{pmatrix}^T.$$
    However, applying Algorithm \ref{alg more improved} to the matrix $\Tilde{G}(z_1,z_2)$ we obtain a realization of dimension $5$ with
    $$\hat{\Psi}(z_1,z_2)=\begin{pmatrix}
        z_1z_2&z_2\\
        &&z_1^2&z_2&z_1
    \end{pmatrix}^T,$$
    i.e.
\begin{align*}
A_1&=\begin{pmatrix}
  0 & 1 & 0 & 0 & 0\\
 0 & 0 & 0 & 0 & 0\\
 0 & 0 & 0 & 0 & 1\\
    0 & 0 & 0 & 0 & 0\\
    0 & 0 & 0 & 0 & 0
\end{pmatrix}, A_2=\begin{pmatrix}
    0 & 0 & 0 & 0 & 0\\
    0 & 0 & 0 & 0 & 0\\
   0 & 0 & 0 & 0 & 0\\
   0 & 0 & 0 & 0 & 0\\
   0 & 0 & 0 & 0 & 0
\end{pmatrix}, B_1=\begin{pmatrix}
   0&0\\
   0 & 0\\
     0 & 0\\
       0 & 0\\
   0 &1
\end{pmatrix}, B_2=\begin{pmatrix}
     0 & 0\\
    1 &0\\
      0 & 0\\
        0 & 1\\
          0 & 0
\end{pmatrix},\\
C&=\begin{pmatrix}
  0 & 0 & 1 & 0 & 0\\
  1 & 1 &  0 & 1 & 0\\
  0 & 1 & 0 & 0 & 0
\end{pmatrix}.\end{align*}
The row-span of $\mathcal{O}$ is equal to $\langle e_3^\top, e_1^\top+e_2^\top+e_4^\top,e_2^\top, e_5^\top, e_1^\top\rangle$
and hence, $\mathcal{O}$ is full rank and this realization is observable.
\end{example}

\section{Conclusion}
In this paper, we introduced an error-correction decoding method for two-dimensional (2D) convolutional codes, called the 2D Viterbi algorithm, which extends the classical Viterbi algorithm to the two-dimensional setting. Since the computational complexity of the decoding algorithm depends directly on the dimension of the linear systems representation, obtaining realizations of small dimension is essential to keep decoding feasible. We presented several methods to reduce the system dimension but it is still an open problem to find a minimal realization of an encoder. In this context, another problem related to minimality concerns the characterization of encoders that admit minimal dimensional realizations among all encoders of a given code. While this problem has been solved in the one-dimensional case, its extension to the two-dimensional setting is expected to be considerably more challenging, as the minimal realizations of a code will be minimal realizations of certain encoders, thus requiring the determination of minimal realizations for each encoder.

\section*{Acknowledgements}

The first and fourth author are supported by the Center for Research and Development in Mathematics and Applications (CIDMA) under the Portuguese Foundation for Science and Technology (FCT, https://ror.org/00snfqn58) Multi-Annual
Financing Program for R\&D Units, grants UID/4106/2025 and UID/PRR/4106/2025.
The second and third author are supported by the German research foundation, project number
513811367.

\bibliographystyle{alpha}
\bibliography{biblio} 

\appendix

\section{Step-by-Step Example of the 2D Viterbi Algorithm (Algorithm \ref{alg 2D viterbi})}

Consider the encoder $G(z_1,z_2)=
\in\mathbb F_2[z_1,z_2]^{2\times 1}$, i.e. $n=2$, $k=1$, $\nu_1=1$.
Calculating a realization for $G(z_1,z_2)$ according to Algorithm \ref{Alg 1}, 
one obtains $A_1=A_2=%
$. 
Easy calculations show that there exists no realization of dimension one.
We assume that the codeword $$c(z_1,z_2)=G(z_1,z_2)\cdot (z_2+1)=%
$$ is sent and that the received word is equal to $r(z_1,z_2)=%
$

The global updating equations are:
\[
X_{i+1}(\xi) = \left( %
 U_i(\xi)
\]
Now, from the received word, we have that:
\[
R_0(\xi) = %
, \quad N=2.
\]
We initialize the algorithm with:
\[
i=0: \ d_{\min} = \infty, \ sp_{\min} = \emptyset, \quad X_0(\xi)=\mathbf{0}: (d=0, sp=\emptyset)
\]
Evaluating the branches reaching the nodes at this stage:
\begin{itemize}

    \item For vertex $X_1(\xi)=%
\right)$ is the label of the arc $\left(%
 \right\}
    \]
    Updating the partial survivor history:
    \[
    \overline{sp} = \emptyset \rightsquigarrow \overline{sp} = \left\{ %
\right)$ is the label of the arc $\left(%
 \right\}
    \]
    Updating its partial survivor history and label:
    \[
    \overline{sp} = \emptyset \rightsquigarrow \overline{sp} = \left\{ %
    }
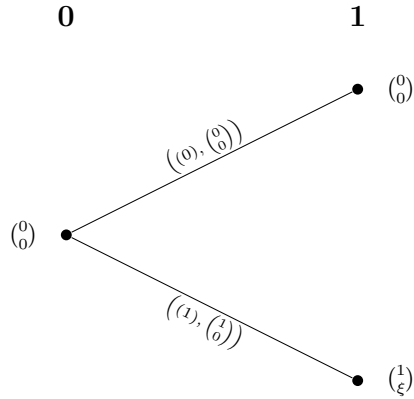
    \caption{Trellis expansion at stage $i=0$.}
    \label{fig:trellis_step_0}
\end{figure}

Next, see Figure \ref{fig:trellis_step_1} to consider:
\[
i=1: \ d_{\min} = \infty, \ sp_{\min} = \emptyset
\]
\begin{itemize}
    \item For vertex $X_2(\xi)=%
\right)$ is the label of the arc $\left(%
\right)$.\\
   
    Hence, we have
   \begin{align*}
    &\text{sum 1: } 1 + d\left( %
 \right\}
    \end{align*}
    Updating the partial survivor history:
    \[
    \overline{sp} = \emptyset \rightsquigarrow \overline{sp} = \left\{ \left( %
\right)$\\
   
    Hence, we have
\begin{align*}
    &\text{sum 1: } 1 + d\left( %
 \right\}
    \end{align*}
        Updating its partial survivor history and label:
    \[
    \overline{sp} = \emptyset \rightsquigarrow \overline{sp} = \left\{ \left( %
\right)$\\
    
    Hence, we have
    \begin{align*}
    &\text{sum 1: } 1 + d\left( %
 \right\}
    \end{align*}
    Updating its partial survivor history and label:
        \[
    \overline{sp} = \emptyset \rightsquigarrow \overline{sp} = \left\{ \left( %
\right)$\\
  
   Hence, we have
   \begin{align*}
    &\text{sum 1: } 1 + d\left( %
 \right\}
    \end{align*}
    Updating its partial survivor history and label:
    \[
    \overline{sp} = \emptyset \rightsquigarrow \overline{sp} = \left\{ \left( %
    }
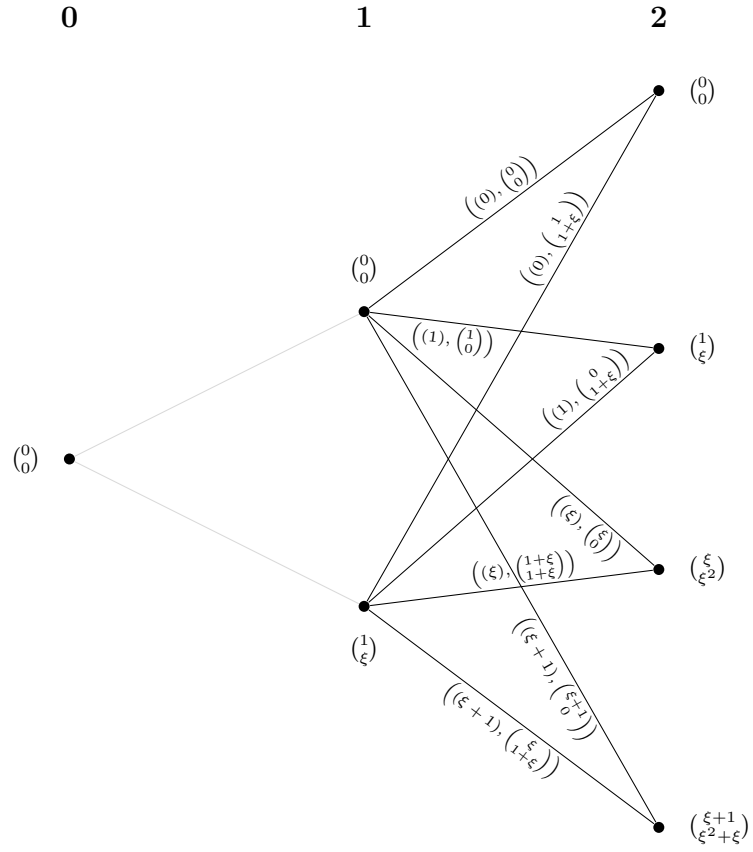
    \caption{Trellis expansion at stage $i=1$.}
    \label{fig:trellis_step_1}
\end{figure}

\newpage

Next, see Figures \ref{fig:trellis_topology_2}, \ref{fig:trellis_state_00}, \ref{fig:trellis_state_1xi}, \ref{fig:trellis_state_xixi2} and \ref{fig:trellis_state_xi1} to consider:
\[
\ i=2: \ d_{\min} = \infty, \ sp_{\min} = \emptyset
\]

\begin{itemize}
    \item For vertex $X_3(\xi)=%
\right)$\\
    Hence, we have
    \begin{align*}
    &\text{sum 1: } 0 + d\left( %
 \right\}
    \end{align*}
    Updating its partial survivor history and label:
    \[
    \overline{sp} = \emptyset \rightsquigarrow \overline{sp} = \left\{ \left( %
 \right) \right\} \right)
    \]
    Since we are at stage $i=2\geq N$, we now set:
    \[
    d_{\min}= 1,\ sp_{\min}=\left\{ \left( %
\right)$\\
    Hence, we have
    \begin{align*}
    &\text{sum 1: } 0 + d\left( %
 \right\}
    \end{align*}
    Updating its partial survivor history and label:
    \[
    \overline{sp} = \emptyset \rightsquigarrow \overline{sp} = \left\{ \left( %
\right)$\\
    Hence, we have
    \begin{align*}
    &\text{sum 1: } 0 + d\left( %
 \right\}
    \end{align*}
    Updating its partial survivor history and label:
    \[
    \overline{sp} = \emptyset \rightsquigarrow \overline{sp} = \left\{ \left( %
\right)$\\
    Hence, we have
    \begin{align*}
    &\text{sum 1: } 0 + d\left( %
 \right\}
    \end{align*}
    Updating its partial survivor history and label:
    \[
    \overline{sp} = \emptyset \rightsquigarrow \overline{sp} = \left\{ \left( %
\right)$\\
    Hence, we have
    \begin{align*}
    &\text{sum 1: } 0 + d\left( %
 \right\}
    \end{align*}
    Updating its partial survivor history and label:
    \[
    \overline{sp} = \emptyset \rightsquigarrow \overline{sp} = \left\{ \left( %
\right)$\\
    Hence, we have
    \begin{align*}
    &\text{sum 1: } 0 + d\left( %
 \right\}
    \end{align*}
    Updating its partial survivor history and label:
    \[
    \overline{sp} = \emptyset \rightsquigarrow \overline{sp} = \left\{ \left( %
\right)$\\
    Hence, we have
    \begin{align*}
    &\text{sum 1: } 0 + d\left( %
 \right\}
    \end{align*}
    Updating its partial survivor history and label:
    \[
    \overline{sp} = \emptyset \rightsquigarrow \overline{sp} = \left\{ \left( %
\right)$\\
    Hence, we have
    \begin{align*}
    &\text{sum 1: } 0 + d\left( %
 \right\}
    \end{align*}
    Updating its partial survivor history and label:
    \[
    \overline{sp} = \emptyset \rightsquigarrow \overline{sp} = \left\{ \left( %
    }
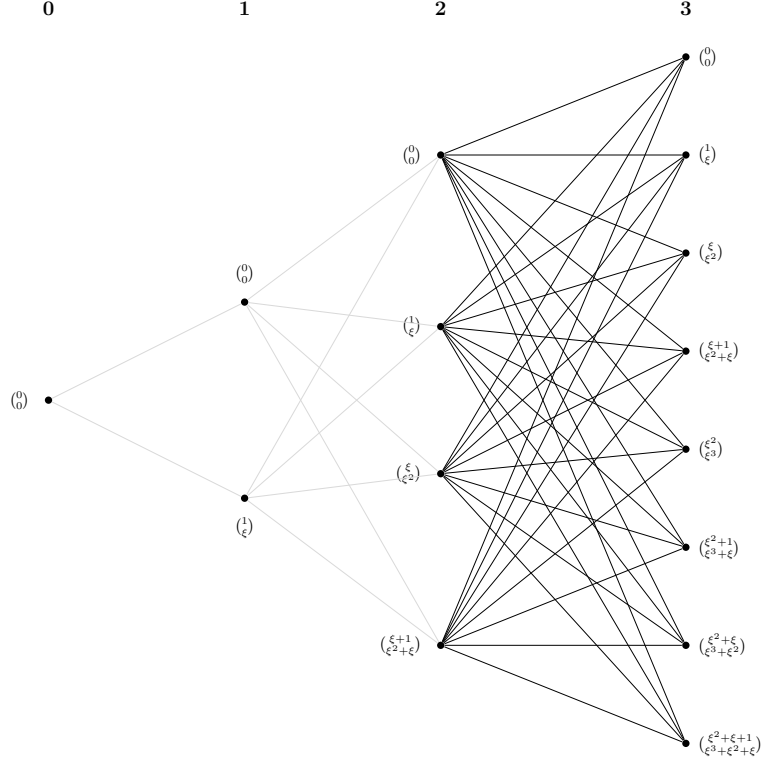
    \caption{Complete Trellis diagram at stage $i=2$.}
    \label{fig:trellis_topology_2}
\end{figure}

For $i=3$, one will get that $\bar d$ is at least $2$ for all the possible vertices $X_4(\xi)$. This means that the algorithm stops with $d_{\min}=1$ and $sp_{\min}=\left\{ \left( %
\right) \right\}$. The output is the list of minimal codewords $c_{\min}$ that correspond to the paths in $sp_{\min}$, in this case it is only one path and corresponding codeword. The labels of this path are 
\begin{align*}
(U_0(\xi),W_0(\xi))&=\left( (1),%
\right)
\end{align*}
and hence, the corresponding codeword is $$W_0(\xi)+W_1(\xi)\eta +W_2(\xi)\eta^2=%
,$$
which is indeed the codeword that was sent. The message corresponding to this codeword is $U_0(\xi)+U_1(\xi)\eta+U_2(\xi)\eta^2=1+\xi\eta=1+z_2$.

The following subfigures illustrate the decoding process.
\newpage

\begin{figure}[htbp]
    \centering
    \resizebox{0.8\textwidth}{!}{
    %
    }
    \caption{Trellis transitions originating from state $X_2(\xi) = %
$ to all destination states at stage $i+1=3$.}
    \label{fig:trellis_state_00}
\end{figure}

\newpage

\begin{figure}[htbp]
    \centering
    \resizebox{0.85\textwidth}{!}{
    %
    }
    \caption{Trellis transitions originating from state $X_2(\xi) = %
$ to all destination states at stage $i+1=3$.}
    \label{fig:trellis_state_1xi}
\end{figure}

\newpage

\begin{figure}[htbp]
    \centering
    \resizebox{0.85\textwidth}{!}{
    %
    }
    \caption{Trellis transitions originating from state $X_2(\xi) = %
$ to all destination states at stage $i+1=3$.}
    \label{fig:trellis_state_xixi2}
\end{figure}

\newpage

\begin{figure}[htbp]
    \centering
    \resizebox{0.85\textwidth}{!}{%
    %
%
    }
    \caption{Trellis transitions originating from state $X_2(\xi) = %
$ to all destination states at stage $i+1=3$.}
    \label{fig:trellis_state_xi1}
\end{figure}

\end{document}